\documentclass[10pt]{article}

\usepackage{amsmath}
\usepackage{amssymb}
\usepackage{amsthm}
\usepackage{color,graphicx}	
\usepackage{float}
\usepackage{url}
\usepackage{colortbl}
\usepackage{pstricks}
\usepackage{rotating}
\usepackage{multirow}
\usepackage{calc}
\usepackage[dvips]{hyperref}

\newtheorem{definition}{Definition}[section]
\newtheorem{corollary}[definition]{Corollary}

\newtheorem{lemma}[definition]{Lemma}

\newtheorem{theorem}[definition]{Theorem}

\numberwithin{equation}{section}

\newcommand{\e}[1]{\mathbf{E}\left[ #1 \right]}

\newcommand{\pr}[1]{\text{Pr}\left( #1 \right)}

\title{%
	Repeated Matching Pennies \\ with Limited Randomness\footnote{This research was partially supported by NSF grants CCF-0829754 and DMS-0652521.}}

\author{Michele Budinich\footnote{This work was done while the author was visiting Northwestern University, Department of Electrical Engineering and Computer Science.} \\
	{\bf \sf m.budinich@imtlucca.it} 
	\and
	Lance Fortnow \\
	{\bf \sf fortnow@eecs.northwestern.edu}
	}

\date{\today}

\begin{document}

\maketitle

\begin{abstract}
We consider a repeated Matching Pennies game in which players
have limited access to randomness. Playing the (unique) Nash
equilibrium in this $n$-stage game requires $n$ random bits.
Can there be Nash equilibria (or $\varepsilon$-Nash equilibria)
that use less than $n$ random coins?

Our main results are as follows
\begin{itemize}
\item We give a full characterization of approximate
    equilibria, showing that, for any $\gamma \in [0,1]$,
    the game has a $\gamma$-Nash equilibrium if and only if
    both players have $(1 - \gamma)n$ random coins.
\item When players are bound to run in polynomial time with $n^\delta$ bits of randomness,
    approximate Nash equilibria can exist if and only if
    one-way functions exist.
\item It is possible to trade-off randomness for running
    time. In particular, under reasonable assumptions, if
    we give one player only $O(\log n)$ random coins but
    allow him to run in arbitrary polynomial time with $n^\delta$ bits of randomness and we
    restrict his opponent to run in time $n^k$, for some
    fixed $k$, then we can sustain an $\varepsilon$-Nash
    equilibrium.
\item When the game is played for an infinite amount of
    rounds with time discounted utilities, under reasonable
    assumptions, we can reduce the amount of randomness
    required to achieve a $\varepsilon$-Nash equilibrium to
    $n^\delta$, where $n$ is the number of random coins
    necessary to achieve an approximate Nash equilibrium in
    the general case.
\end{itemize}
\end{abstract}


\newpage
\section{Introduction}
In the classical setting of Game Theory, one of the core
assumptions is that all participating agents are ``fully
rational''. This amounts not only to the fact that an agent
must be able to make optimal decisions, given the other
players' actions, but also to the fact that he must understand
how these actions will affect the behavior of all other
participants. If this is the case, a Nash equilibrium can be
viewed as a set of strategies in which each agent is simply
computing his best response given his opponents' actions.
However, in real world strategic interactions, people often
behave in manners that are not fully rational. There are many
reasons behind non-rational behavior, we focus on two:
limitations on computation and limitations on randomness.

Since the work of Herbert Simon~\cite{simon1955behavioral},
much research has focused on defining models that take
computational issues into account. In recent years, the idea
that the full rationality assumption is often unrealistic has
been formalized using tools and ideas from computational
complexity. It is in fact easy to come up with settings, as in
Fortnow and Santhanam \cite{fortnow2010bounding}, in which
simply computing a best response strategy involves solving a
computationally hard problem. Furthermore there is strong
evidence that, in general, the problem of finding a Nash
equilibrium is computationally difficult for matrix games
(Daskalakis, Goldberg and
Papadimitriou \cite{DBLP:journals/siamcomp/DaskalakisGP09}, Chen and
Deng \cite{chen2006settling}).

Traditionally bounded rationality has focused on two
computational resources: time and space. In this paper we focus
on another fundamental resource: randomness.

It is a basic fact that games in which agents are not allowed
to randomize might have no Nash equilibrium. In this sense,
randomness is essential in game theory. We focus on a simple
two player zero-sum game that captures this: Matching Pennies
(Figure \ref{fig:game}).
      \begin{figure}[h]
      \begin{center}
      \includegraphics[width=1.8in]{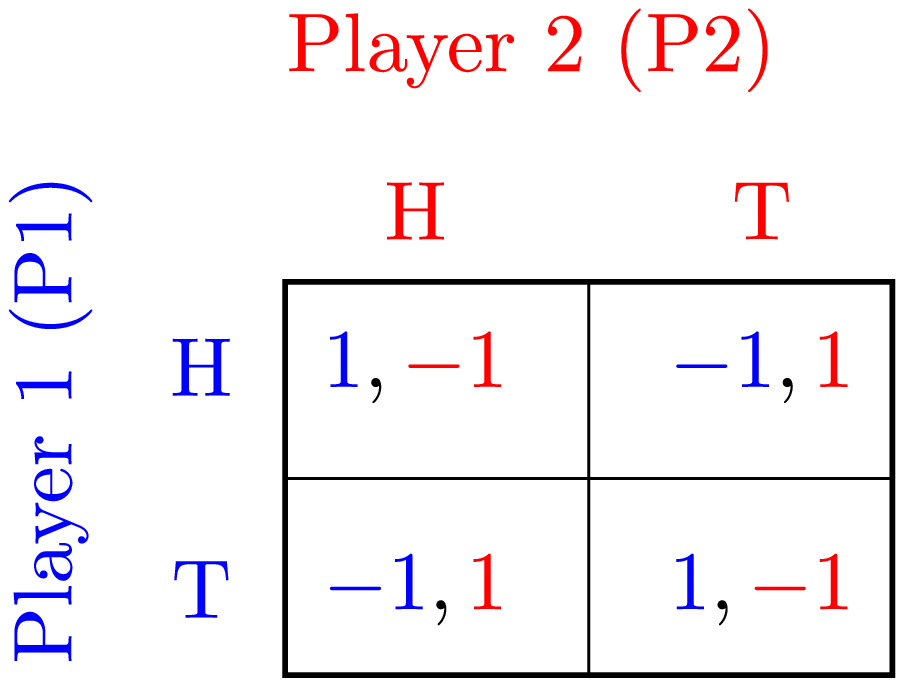}
      \end{center}
      \caption{The payoff bimatrix for the Matching Pennies game.}
      \label{fig:game}
      \end{figure}
Specifically, we consider the repeated version of Matching
Pennies, played for $n$ rounds. In this game, the unique Nash
equilibrium is the one in which, at every round, both players
choose one of their two strategies uniformly at random. The
algorithm that implements this strategy requires $n$ random
coins, one for each round. The main question we address in this
paper is: can there be Nash equilibria if the amount of
randomness available to both players is less than $n$?

First we show that, in general, the game cannot have a Nash
equilibrium in which both players have only a fraction of $n$
random coins. In particular, when we give both players $n(1
-\gamma)$ random coins, we can only achieve a $\gamma$-Nash
equilibrium. This turns out to be tight, in the sense that we
can show that any game with a $\gamma$-Nash equilibrium both
players must have at least $n(1 - \gamma)$ coins.

The proof of this fact, however, relies on the players' ability
to implement a strategy that runs in exponential time. We then
consider games in which the players' strategies are
polynomially-bounded. Using ideas developed in cryptography and
computational complexity we show that, in this setting,
$\varepsilon$-Nash equilibria that use only $n^{\delta}$ coins
exist if and only if one-way functions exist.

We also show that the amount of randomness can be ``traded''
for time. If we allow one of the players to run in arbitrary
polynomial time, but use only $O(\log n)$ bits, we can still
achieve a $\varepsilon$-Nash equilibrium if we restrict his
opponent to run in time $n^k$ for some fixed $k > 0$, while
giving him $n^{\delta}$ random bits.

Finally we consider an infinitely repeated game with time
discounted utilities. In this case, in general, for any
discount factor $\delta$ and approximation $\varepsilon$, we
can always achieve a $\varepsilon$-Nash equilibrium with only
$n$ random coins, if $n$ is large enough. When we limit
players' strategies to polynomial size circuits, we can reduce
the amount of randomness to $n^{\delta}$, for any $\delta > 0$.

\subsubsection*{Related Work}
There are many recent approaches to bounded rationality using a
computational complexity perspective. For instance Halpern and
Pass \cite{halpernpass2010} study games in which players'
strategies are Turing machines. The idea of considering
randomness as a costly resource in game theory has received
only limited attention. Kalyanaraman and Umans
\cite{kalyanaraman2007algorithms} study zero-sum games, and
give both an efficient deterministic algorithm for finding
$\varepsilon$-Nash equilibria, as well as a weaker, but more
general, result in the spirit of our Lemma
\ref{lemma:unbounded1}, giving a randomness-efficient adaptive
on-line algorithm for playing repeated zero-sum games. Hu
\cite{hu2009complexity} also considers a similar setting but he
is concerned with computability rather than complexity. He
considers infinitely repeated plays of 2 player zero-sum games
that have no pure strategy Nash equilibrium, and in which
players have a set of feasible actions, which represents both
the strategies they can play and the strategies they can
predict. In this setting Hu gives necessary and sufficient
conditions for the existence of Nash equilibria. Finally
Gossner and Tomala \cite{gossner2008entropy}, give entropy
bounds on Bayesian learning in a game theoretic setting, in a
more general framework then this paper. Their results applied to
Matching Pennies do not achieve the tight bounds we get in
Lemma~\ref{lemma:unbounded1}.
\\

The rest of the paper is organized as follows. In Section
\ref{sec:background} we introduce the notation and known
results used. Section \ref{sec:info} presents an information
theoretic impossibility result. Section \ref{sec:eff} considers
players whose strategies are limited to polynomial sized
Boolean circuit families, while Sections \ref{sec:nw} and
\ref{sec:infinite} give extensions of the main results to
complexity pseudorandom number generators and infinitely repeated versions of the game.

\section{Background and Definitions} \label{sec:background}

\subsection{Game Theory Notation}
Throughout the paper we consider a repeated game of Matching
Pennies. We focus on this game because it captures one of the
fundamental aspects of randomness in game theory. Studying such
a simple game also allows us to get tight bounds. However,
variations of our results extend to other similar 2 person
zero-sum repeated games.

The payoffs at each round are shown in Figure \ref{fig:game}.
Let $h: \{\text{H,T}\} \times \{\text{H,T}\} \to \{-1, 1\}$, be
the payoff to Player 1 (P1), and $-h$ the payoff for P2. When
we allow the players to randomize, we denote as $S_i$ a
randomized strategy on $\Delta\{H, T\}$ for player $i$. P1's
expected payoff in one round is
	\[
	\e{h(S_1, S_2)} = \sum_{s_1,s_2 \in \{\text{H,T}\}} \pr{ S_1 = s_1} \pr{S_2 = s_2} h(s_1, s_2).
	\]
Let $u:\{\text{H,T\}}^n \times \{\text{H,T\}}^n \to \{-n,
\cdots, n\}$ be P1's cumulative payoff when the game is played
for $n$ rounds. In the repeated game, mixed strategies can be
viewed as distribution over sequences of length $n$ that are dependend on the opponent's strategy. Given the adversary's strategy, let $R_i =
(r_i^1, \cdots, r_i^n) \in \Delta\{\text{H,T\}}^n$ denote a
randomized strategies for player $i$. P1's expected cumulative
payoff is
	\[
	\e{u(R_1, R_2)} = \sum_{t = 1}^n \e{h(r_1^t, r_2^t)}.
	\]
Finally we define the expected average payoff to P1 for the
$n$-round game as
	\[
	\e{U(R_1, R_2)} = \frac{\e{u(R_1, R_2)}}{n},
	\]
and consequently P2's expected payoff is $-\e{U(R_1, R_2)}$. To
denote player's $i$ payoff we will sometimes use the standard
notation $\e{U(R_i, R_{-i})}$, where $R_i$ is player's $i$
mixed strategy and $R_{-i}$ is his opponent's mixed strategy.

\begin{definition}[Nash equilibrium]
A pair of mixed strategies $(R_1, R_2)$ is a Nash equilibrium
for the $n$-stage Matching Pennies game if, for $i = 1,2$:
	\[
	\e{U(R_i, R_{-i})} \ge \e{U(R_i', R_{-i})} \quad \text{for all } R_i' \in \Delta\{\text{H,T}\}^n.
	\]
\end{definition}

In some cases we will consider a relaxed notion of equilibrium,
namely $\varepsilon$-Nash equilibrium.

\begin{definition}[$\varepsilon$-Nash equilibrium]
A pair of mixed strategies $(R_1, R_2)$ is a $\varepsilon$-Nash
equilibrium for the $n$-stage Matching Pennies game if, for $i
= 1,2$:
	\[
	\e{U(R_i, R_{-i})} \ge \e{U(R_i', R_{-i})} - \varepsilon \quad \text{for all } R_i' \in \Delta\{\text{H,T}\}^n.
	\]
\end{definition}

\subsection{Complexity and Pseudorandomness}
We give a brief description of the pseudorandomness tools we
need for this paper. For more details we recommend the
textbooks of Arora and Barak~\cite{Arora:1253363} and
Goldreich~\cite{goldreich2004foundations}.

The model of computation used throughout most of the paper is
based on Boolean Circuits. We consider circuits with
\textsf{AND, OR} and \textsf{NOT} gates, and denote by $C_n$ a
circuit with $n$ input nodes. A circuit family $\{C_i\}_{i \in
\mathbb{N}}$ is an infinite collection of circuits, intuitively
one for each input length.

The size of a circuit $|C_n|$ is the number of gates. A circuit
family is polynomial sized if there is a $k >0$ such that, for
all $n$, $|C_n| \le n^k$. The class of languages recognizable
by families of polynomial sized circuits is called
\textsf{P/Poly}. Any language that can be decided in polynomial
time by a deterministic or randomized Turing machine is also in
\textsf{P/Poly}. Formally $\textsf{P} \subseteq \textsf{BPP}
\subseteq \textsf{P/Poly}$.

A function is one-way if it is easy to compute and hard to
invert.
\begin{definition}[One-way Function]
A one-way function is a polynomial-time computable function
$f:\{0,1\}^n \to \{0,1\}^n$ such that, for all polynomial size
circuits $D$ and $y = f(x)$ (where $x$ is chosen uniformly at
random on $\{0,1\}^n$), $\Pr(f(D(y)) = f(x)) < n^{-c}$ for all
$c
> 0$ and sufficiently large $n$.
\end{definition}

Informally, two objects are indistinguishable if no polynomial
sized circuit family can tell them apart with noticeable
probability.

\begin{definition}[Indistinguishability]
Let $X,Y$ be two random variables on $\{0,1\}^n$. We say that
$X$ and $Y$ are computationally indistinguishable if for every
family of polynomial size circuits $\left\{ C_i \right\}_{i \in
\mathbb{N}}$, every $c > 0$ and for sufficiently large $n$
	\[
	\Big| \pr{C_n\left( X \right) = 1} - \pr{C_n\left( Y \right) = 1} \Big| < \frac{1}{n^c}.
	\]
\end{definition}

A cryptographic pseudorandom number generator (PRNG) is a
deterministic algorithm whose output is computationally
indistinguishable from the uniform distribution, provided that
it's input is truly random. We will denote by $U_k$ a random
variable uniformly distributed on $\{0,1\}^k$.

\begin{definition}[Cryptographic PRNG]
A cryptographic pseudorandom number generator is a
deterministic polynomial time algorithm $G: \{0,1\}^{l(n)} \to
\{0,1\}^{n}$, where $l(n)<n$ is a polynomial time computable
function, such that $G(U_{l(n)})$ and $U_{n}$ are
computationally indistinguishable.
\end{definition}

There are two basic properties about pseudorandom number generators that we will use.
One relates the notion of pseudorandomness to the notion of
predictability.

\begin{definition}[Unpredictable]
Let $G: \{0,1\}^{l(n)} \to \{0,1\}^{n}$ be a polynomial time
algorithm, and $G(x) = (y_1, \cdots, y_n)$. We call $G$
unpredictable if for every family of polynomial size circuits
$\left\{ D_i \right\}_{i \in \mathbb{N}}$, all $c>0$ and for
sufficiently large $n$
	\[
	\pr{  D_i\left( y_1, \dots, y_{i-1} \right) = y_{i} } \le \frac{1}{2} + \frac{1}{n^c}.
	\]
\end{definition}

Intuitively a pseudorandom number generator must be unpredictable, otherwise we could
easily build a test for it by using the predictor circuits. In
1982 Yao \cite{yao1982theory} proved the opposite implication,
thus establishing the following theorem.

\begin{theorem}[Yao's Theorem] \label{yaothm}
A polynomial time algorithm $G: \{0,1\}^{l(n)} \to \{0,1\}^{n}$
is unpredictable if and only if $G$ is a pseudorandom number generator.
\end{theorem}

H\r{a}stad, Impagliazzo, Luby and Levin in 1999
\cite{hastad1999pseudorandom} showed how to construct pseudorandom number generators
with polynomial expansion based on one-way functions.

\begin{theorem}[PRNG's from one-way functions] \label{hillthm}
One way functions exist if and only if for every $\delta >0$
there is a pseudorandom number generator with $l(n) = n^\delta$.
\end{theorem}

Cryptographic pseudorandom number generators' main power lies in the ability to fool any
polynomial sized adversary, while running in polynomial time.
However, in other areas of complexity, such as derandomization,
the crucial issue is having a smaller seed.

\begin{definition}[Complexity PRNG]
A complexity pseudorandom number generator is a $2^{l(n)}$ time
computable function $G: \{0,1\}^{l(n)} \to \{0,1\}^{n}$, such
that for any circuit $C$ of size $n$
	\[
	\Big| \pr{C\left( U_{l(n)} \right) = 1} -
	      \pr{C\left(G\left( U_n \right)\right) = 1} \Big| < n^{-1}.
	 \]
\end{definition}

The essential difference with cryptographic pseudorandom number generators is the order
of quantifiers. A cryptographic pseudorandom number generator fools circuits of an
arbitrary polynomial size. The complexity pseudorandom number generator fools circuits
only of a fixed polynomial size but under the right assumptions
requires far fewer random bits.

Impagliazzo and
Wigderson~\cite{DBLP:journals/jcss/ImpagliazzoW01} building on
a series of paper starting with Nisan and
Wigderson~\cite{nisan1994hardness} characterize when complexity
pseudorandom number generators exist.

\begin{theorem} \label{thm:nw}
There exists $L \in \textsf{DTIME}(2^{O(l(n))})$ and
$\varepsilon > 0$ such that no circuit of size at most
$2^{\varepsilon l(n)}$ can compute $L$ if and only if there
exists a complexity pseudorandom number generator with $l(n) = C \log n$ for some $C >
0$.
\end{theorem}

\section{Information Theoretic Bounds} \label{sec:info}
In this section we make no computational assumptions on the
players, and show that there can be no Nash equilibrium if we
limit the amount of randomness available to both players.

\begin{lemma} \label{lemma:unbounded1}
For any $\gamma \in [0,1]$, if P2 has less than $n(1-\gamma)$
random bits, then P1 has a deterministic strategy $A$ that
achieves an expected average payoff of at least $\gamma$.
\end{lemma}

\begin{proof}
We will give a strategy $A = (a^1, \cdots, a^n)$ for P1 that
achieves a high payoff against any strategy $B = (b^1, \cdots,
b^n)$ from P2.

Player 1 will enumerate all of P2's possible coin flips, and
will obtain a set of $2^{n(1-\gamma)}$ possible strategies, one
of which is the one being used by P2. After each play by P2, P1
can eliminate all the strategies that do not play that action
at that round. Let $S^t$ be the set of strategies that are
consistent with P2's plays up to round $t$. Initially the set
$S^1$ contains $2^{n(1 - \gamma)}$ strategies, and, for all
$t$, $|S^{t+1}| \le |S^t|$.

The strategy $A$ for P1 is straightforward: at round $t$, P1
will play based on the most likely event: he will consider all
strategies in $S^t$ and play H if the majority of strategies in
$S^t$ use H at round $t$ and play T otherwise. Let $p^t$ be the
exact fraction of strategies that are the majority at round
$t$,
	\begin{equation} \label{eq:pt}
	p^t = \frac{\max\{ |\{b^t ~|~ b^t = H\}|,|\{b^t ~|~ b^t = T\}| \}}{|S^t|},
	\end{equation}
so that $p^t \in [1/2, 1]$. P1's expected payoff at round $t$
is:
	\[
	\e{h(a^t, b^t)} = p^t - (1 - p^t) = 2 p^t - 1 \ge 0.
	\]
Thus P1's average expected payoff is at least 0. To show that
P1 can actually achieve an average expected payoff of $\gamma$
we need to consider the amount of information P1 gains at each
round. We define the following potential function $\phi: \{1,
\dots, n\} \to \mathbb{R}$:
	\[
	\phi(t) = \sum_{k = 1}^{t-1} h(a^k, b^k) - \log |S^t|,
	\]
which considers both the accumulated payoff for P1 and the
log-size of the set of consistent strategies. At time $t = 1$
there are $2^{n(1 - \gamma)}$ possible strategies for P2, so
$\phi(1) = -n(1 - \gamma)$. We will now lower bound the
expected increase in $\phi$ at each round. We can express this
as
	\begin{align*}
	\e{\phi(t+1) - \phi(t)} &= \left( \e{ \sum_{k = 1}^{t} h(a^k, b^k)} - \e{\sum_{k = 1}^{t-1} h(a^k, b^k)}\right) \\
	& \qquad \qquad- \left( \e{ \log |S^{t+1}| } - \e{\log|S^t|} \right) \\
	&= \e{h(a^t, b^t)} - \left( \e{\log |S^{t+1}| - \log|S^t|} \right) \\
	&= 2p^t - 1 - \left( \e{\log |S^{t+1}| - \log|S^t|} \right).
	\end{align*}
Now consider $\e{\log|S^{t+1}|}$. When P1 looses he can
eliminate a $p^t$ fraction of strategies, thus the new set
$S^{t+1}$ will contain a $(1 - p^t)$ fraction of the strategies
in $S^t$. On the other hand, when P1 wins, $|S^{t+1}| =
p^t|S^t|$. To complete the analysis we have to consider two
cases, since if $p^t = 1$ then $\e{\log|S^{t+1}|}$ is not well
defined.

First assume $p^t = 1$. This happens when all feasible
strategies for P2 have the same action at round $t$. In this
case P1 will win with probability 1, and the size of the set of
feasible strategies will stay the same. So, overall, the
increase in $\phi$ will be 1.

Now assume $p^t \in [1/2, 1)$. Then, the expected size of the
set $S^{t+1}$ is:
	\begin{equation} \label{eq:st+1}
	\e{\log |S^{t+1}|} = \left(  (1 - p^t) \log (1 - p^t) |S^t| + p^t \log p^t |S^t| \right).
	\end{equation}
The expected change in $\log |S^t|$ does not depend on $S^t$,
but only on $p^t$, since
	\begin{align*}
	\e{\log |S^{t+1}| - \log |S^{t}|} &= \left[ (1 - p^t) \log (1 - p^t) |S^t| \right. \\
		& \left. \quad \qquad + p^t \log p^t|S^t| \right] - \log |S^t| \\
	&= (1 - p^t) \log (1-p^t) +  p^t \log p^t.
	\end{align*}
So that the overall change in potential when $p^t < 1$ is
	\[
	\e{\phi(t+1) - \phi(t)} \ge
	2p^t - 1 - ( 1 - p^t) \log (1-p^t) -  p^t \log p^t,
	\]
which is always at least $1$ for $p^t \ge 1/2$.

So, for all $p^t \in [1/2, 1]$, at each step the potential
function $\varphi$ increases by at least $1$. Thus, after $n$
rounds we have that
	\begin{align*}
	\e{\phi(n)} &\ge \phi(1) + \min_t \{ n \e{ \phi(t+1) - \phi(t) } \} \\
	&\ge -n(1 - \gamma) + n = n\gamma.
	\end{align*}
Since P1's expected payoff is at least $\e{\phi(n)}$, this
completes the proof.
\end{proof}

This result immediately implies that, without any computational
assumption, there can be no equilibrium with less than $n$
random coins.

\begin{corollary}
For all $\gamma_1, \gamma_2 \in [0,1]$ such that $\gamma_1 +
\gamma_2 > 0$, if P1 and P2 have, respectively, $n(1 -
\gamma_1)$ and $n(1 - \gamma_2)$ random coins, then there can
be no Nash equilibrium in the $n$-stage Matching Pennies
repeated game.
\end{corollary}

\begin{proof}
Assume, by contradiction, that $(S_1, S_2)$ is such a Nash
equilibrium. By Lemma \ref{lemma:unbounded1}, P1's expected
payoff must be $\e{U(S_1, S_2)} \ge \gamma_2$ and P2's payoff
$-\e{U(S_1, S_2)} \ge \gamma_1$, otherwise they would be better
off by using the majority strategy in the proof of Lemma
\ref{lemma:unbounded1}. Summing the two inequalities we get
$\gamma_1 + \gamma_2 \le 0$, a contradiction, since we assume
$\gamma_1 + \gamma_2 > 0$.
\end{proof}

However, if we limit the amount of randomness available to both
players, we are still able to achieve an $\varepsilon$-Nash
equilibrium. Furthermore, if the game has a $\varepsilon$-Nash
equilibrium, then both players must have at least $(1 -
\varepsilon)n$ random coins.


\begin{theorem} \label{lemma:delta-eq-tight}
Let $\gamma \in [0,1]$. The game has a $\gamma$-Nash
equilibrium if and only if both players have $n(1 - \gamma)$
random coins.
\end{theorem}

For simplicity we assume $\gamma$ is the same for both players, however a similar result holds even in the case where the two players have a different amount of random coins.

\begin{proof}
To show the ``only-if'' implication, consider, by way of
contradiction, a game that has a $\gamma$-Nash equilibrium
$(S_1, S_2)$ but in which both players have less than $n(1 -
\gamma)$ random bits. Thus there must be a $\gamma' > \gamma$
such that they have exactly $n(1 - \gamma')$ random bits. Since
$(S_1, S_2)$ is a $\gamma$-Nash equilibrium, it must be the
case that, for any strategy $S_1'$ for P1
	\[
	\e{U(S_1, S_2)} \ge \e{U(S_1', S_2)} - \gamma.
	\]
By Lemma \ref{lemma:unbounded1} we know that both players have
a strategy that achieves a payoff of at least $\gamma' >
\gamma$, so that the above implies $\e{U(S_1, S_2)} > 0$.
Applying the same argument to P2, we get $-\e{U(S_1, S_2)} >
0$, a contradiction.

The ``if'' part follows from Lemma~\ref{lemma:delta-eq} below.
\end{proof}

\begin{lemma} \label{lemma:delta-eq}
Let $\gamma \in [0,1]$. If both players have $n(1 - \gamma)$
random coins, then the game has a $\gamma$-Nash equilibrium.
\end{lemma}

For simplicity we assume $\gamma n$ is even.

\begin{proof}
Consider the following strategies: both player use their random
coins to play uniformly at random for the first $n(1 - \gamma)$
rounds. Thereafter P1 will always play $H$, while P2 will
alternate between $H$ and $T$, playing $H, T, \dots$. We claim
that this is a $\gamma$-Nash equilibrium.

First notice that no player can improve his payoff in the first
$n(1 - \gamma)$ rounds, given his opponent's strategy. Let's
consider the remaining $n\gamma$ rounds. P1 could improve his
payoff by playing $H,T,H,T\dots$, however this only increases
his payoff by $\gamma$. This holds also for P2, that could play
$T, T, \dots$, however gaining only $\gamma$.
\end{proof}

\section{Computationally Efficient Players} \label{sec:eff}
The proof of Lemma \ref{lemma:unbounded1} in the previous
section relies heavily the fact that we make no computational
assumptions. In particular, to implement the majority strategy
and compute $p^t$ in \eqref{eq:pt} requires solving $\#$\textsf{P} hard
problems, by reduction from $\#$\textsf{SAT}. If we restrict the players to run in time polynomial
in $n$ this particular strategy likely becomes unfeasible. In
this setting, under reasonable complexity assumptions, it is
possible to greatly reduce the amount of randomness and, at the
same time, achieve a $\varepsilon$-Nash equilibrium.

We consider players' whose actions are polynomial size Boolean
circuits. A strategy is thus a circuit family $\{C_i\}_{i \in
\mathbb{N}}$, such that circuit $C_{l(n)}$ takes as input
$l(n)$ random coins and outputs the $n$ actions to be played.
Notice that this definition implies that each agent can
simulate any of his opponent's strategies.

We consider equilibria that use $n^{\delta}$ random coins for
any $\delta >0$. Theorem \ref{thm:nash-iff-one-way} shows that
such $\varepsilon$-Nash equilibria exist if and only if one-way
functions exist.

\begin{theorem} \label{thm:nash-iff-one-way}
If players are bound to run in time polynomial in $n$, then, for all $\delta > 0$ and sufficiently large $n$,
$\varepsilon$-Nash equilibria that use only $n^\delta$ random
coins exist, where $\varepsilon = n^{-k}$ for all $k > 0$ and
sufficiently large $n$'s, if and only if one-way functions
exist.
\end{theorem}

\begin{proof}
The if part is Lemma \ref{one-way-to-nash}, while the only-if
part is Lemma \ref{nash-to-one-way}.
\end{proof}

As a preliminary result we show that, in our setting, the
expected utility when at least one player uses a pseudorandom number generator can't be
too far from the expected utility when playing uniformly at
random.

\begin{lemma} \label{lemma}
Assume one-way functions exist, and let $G$ be the strategy
corresponding to the output of a pseudorandom number generator. For any strategy $S$ that runs in time polynomial in $n$,
for all $k > 0$ and sufficiently large $n$,
	\[
	\Big|\e{U(G, S)} \Big| \le n^{-k} \text{ and } \Big|\e{U(S, G)} \Big| \le n^{-k}.
	\]
\end{lemma}

\begin{proof}
We prove only the first inequality, the proof for the second
one being symmetric.

Proof by contradiction. Assuming there is a $k > 0$ such that
$\Big|\e{U(G, S)}\Big|> n^{-k}$ for infinitely many $n$'s, we
will construct a test $T$ for $G$, and show that
	\begin{equation}
	\Big|
		\pr{T\left( G\left( U_{l(n)} \right) \right) = 1}   -
		\pr{T\left( U_n \right) = 1}\Big| > n^{-c}
	\label{eq:break}
	\end{equation}
for some $c > 0$ and infinitely many $n$'s, thus contradicting
the assumption that $G$ is a pseudorandom number generator.

First consider the $n$ random variables $A_i(G, S)$, for $i = 1, \cdots, n$, where $A_i(G, S)$ is simply P1's
payoff at round $i$. Since $\sum_{i = 1}^n A_i(G, S) = n U(G, S)$,
	\[
	\sum_{i = 1}^n \e{A_i(G, S)} > n^{1 - k}.
	\]
This implies that there must be an $i$ such that $\e{A_i(G, S)}
> n^{-k}$. Fix that $i$.

The test $T$ takes as input an $n$-bit sequence $x$ and
generates a sequence of plays $s$ according to strategy $S$.
Now $T$ simulates an $n$-stage repeated Matching Pennies game
with strategies $(x,s)$. If P1 wins the $i$-th round then it
will output 1, otherwise the output will be 0. In other words,
$T$ outputs 1 if and only if $A_i(x,s) = 1$. Notice that $T$ runs in time polynomial in $n$.

When $x$ is drawn from the uniform distribution, P1 will win
with probability $1/2$, or $\pr{T(U_n) = 1} = 1/2$.

Now notice, that since $A_i \in \{-1, 1\}$, $\pr{A_i = 1} =
\frac{\e{A_i} + 1}{2}$. This implies that
	\[
	\pr{T(G(U_{l(n)})) = 1} = \frac{\e{A_i} + 1}{2} > \frac{1}{2n^{k}} + \frac{1}{2}.
	\]
Thus
	\[
		\Big|
		\pr{T\left( G_1\left( U_{l(n)} \right) \right) = 1}   -
		\pr{T\left( U_n \right) = 1}\Big| > \frac{1}{2n^{k}},
	\]
which proves the lemma.
\end{proof}

\begin{lemma} \label{one-way-to-nash}
If one-way functions exist and players are bound to run in time polynomial in $n$, then for every $\delta, k > 0$ and
for sufficiently large $n$, the $n$-stage has an $n^{-k}$-Nash
equilibrium in which each player uses at most $n^\delta$ random
bits.
\end{lemma}
\begin{proof}
Assume, by contradiction, that one-way functions exist but
there are values $\delta > 0$ and $k > 0$ such that the game
has no $n^{-k}$-Nash equilibrium in which players use at most
$n^\delta$ random coins.

Since we assume one-way functions exist, by Theorem
\ref{hillthm} there exist pseudorandom number generators that use $n^\delta$ coins.
Assume both players use the output of such pseudorandom number generators as their
strategies (which we call, respectively, $G_1$ and $G_2$).
Since we are assuming that this is not a $n^{-k}$-Nash
equilibrium, one of the players, say Player 1, must have a
strategy $A$ such that
	\begin{equation} \label{eq:nonash}
	\e{U\left( A, G_2 \right)} >  \e{U\left( G_1, G_2 \right)} + n^{-k}.
	\end{equation}
By Lemma \ref{lemma} we can choose a $k' > 0$ such that
	\begin{equation} \label{eq:nonash1}
	\e{ U\left( A, G_2 \right)} > -n^{-k'} + n^{-k}.
	\end{equation}
Pick $c = k' = k+1$, so that $\e{ U\left( A, G_2 \right)} >
n^{-c}$ for $n > 2$. This contradicts Lemma \ref{lemma},
proving the claim.
\end{proof}

We now prove the opposite direction, that is that the existence
of Nash equilibria that use few random bits implies the
existence of one-way functions.

\begin{lemma} \label{nash-to-one-way}
If for every $\delta > 0$ there is a Nash equilibrium in which
each player uses $n^\delta$ random bits and runs in time polynomial in $n$, then one-way functions exist.
\end{lemma}
\begin{proof}
Let $(A,B)$ be such a Nash equilibrium and assume, by
contradiction, that one-way functions don't exist. This implies
that pseudorandom number generators can't exist (Goldreich
\cite{goldreich2004foundations}), and so, $A$ and $B$ can't be
sequences that are computationally indistinguishable from
uniform.

Thus, by Yao's theorem (Theorem \ref{yaothm}), we know that
there are polynomial size circuit families $\left\{ C_i
\right\}_{i \in \mathbb{N}}$ and $\left\{ D_i \right\}_{i \in
\mathbb{N}}$ such that
	\begin{align*}
	\pr{C_i(y_1, \dots, y_{i}) = y_{i+1}} &> 1/2 + \delta_1 \\
	\pr{D_i(z_1, \dots, z_{i}) = z_{i+1}} &> 1/2 + \delta_2,
	\end{align*}
for some $\delta_1, \delta_2 > 0$, where $A(x_1, \dots,
x_{l_1(n)}) = (y_1, y_2, \dots, y_{n})$ and $B(x_1, \dots,
x_{l_2(n)}) = (z_1, z_2, \dots, z_{n})$.

To get a contradiction it is sufficient to show that players
are better off by using the predictor circuits $C$ and $D$.
Consider Player 1: using $D$, at each round he can guess, given
the previous history, the opponent's next move with probability
$\frac{1}{2} + \delta_1$.  Thus his expected payoff at any
round $t$ is
	\[
	\e{h(d^t, b^t)} > \frac{1}{2} + \delta_1 - \frac{1}{2} + \delta_1 = 2\delta_1,
	\]
where the expectation is over the internal coin tosses of the
predictor circuit $D$. The overall expected payoff is
	\[
	\e{U(D, B)} > \frac{1}{n} \sum_{t = 1}^n 2\delta = 2\delta_1.
	\]
Now, let $w = \e{U(A, B)}$ be the value of the expected payoff
when players play $(A,B)$. Consider the following cases:
	\begin{enumerate}
	\item $w \le 0$: this implies that Player 1 could gain
$2\delta_1$ by using strategy $D$,
	\item $w > 0$: by definition Player 2's expected payoff
is $-\e{U(A, B)} < 0$, so Player 2 can achieve a higher
payoff by using his predictor circuit $C$,
	\end{enumerate} In both cases we see that $(A, B)$
can't be a Nash equilibrium, a contradiction.
\end{proof}

\section{Exchanging Time for Randomness} \label{sec:nw}
In this section we determine conditions under which a
$\varepsilon$-Nash equilibrium can arise, given that one of the
players has only a logarithmic amount of randomness and his
opponent must run in time $n^k$ for some fixed $k$. This shows
how we can trade off randomness for time; the player with
$O(\log n)$ random bits runs in time polynomial in $n$, while
the player with more random bits runs in fixed polynomial time.

\begin{theorem} \label{thm:logn-vs-poly}
Assume there exists $f \in \textsf{DTIME}(2^{O(l(n))})$ and
$\varepsilon > 0$ such that no circuit of size at most
$2^{\varepsilon l(n)}$ can compute $f$ and that one-way
functions exist. Let Player 1's strategies be circuits of size
at most $n^k$ that use at most $n^{\delta}$ random bits for
some $k > 2 + c \delta$, where $c \ge 1$ is a constant related
to the implementation of a cryptographic pseudorandom number generator. Assume Player 2
has access to only $M \log n$ random bits. As long as $M > Ck$,
where $C$ is the constant in Theorem \ref{thm:nw}, then for all
$\varepsilon > 0$ and sufficiently large $n$ there is a
$\varepsilon$-Nash equilibrium.
\end{theorem}

\begin{proof}
Let $G_1$ be the cryptographic pseudorandom number generator available to Player 1 and
$G_2$ be the complexity pseudorandom number generator used by player 2. Furthermore let
$\mathcal{S}_1$ be the set of all possible strategies for P1
(for all $\delta > 0$ circuits of size at most $n^k$ that use
$n^{\delta}$ random bits), and $\mathcal{S}_2$ the set of
strategies available to P2 (polynomial size circuit families
and $M \log n$ random coins). We will show that for all
$\varepsilon > 0$ and sufficiently large $n$, $(G_1, G_2)$ is a
$\varepsilon$-Nash equilibrium with the required properties.

First we argue that, for all $\gamma > 0$ and sufficiently
large $n$, $|\e{U(G_1, G_2)}| < \gamma$. The proof of this fact
is similar to the proof of Lemma \ref{lemma}, showing by way of
contradiction, that if $|\e{U(G_1, G_2)}| \ge \gamma$ then we
can build a test for the cryptographic pseudorandom number generator $G_1$.

Now we show that, for the appropriate setting of parameters,
$G_1$ fools $\mathcal{S}_2$ and $G_2$ fools $\mathcal{S}_1$.
For any $k$, Player 2 can fool circuits of size $n^k$ by using
$C \log n^k = Ck \log n$ random bits. So, for $M > Ck$, $G_2$
fools $\mathcal{S}_1$. Notice also that since Player 2 runs in
time $O(n^{Ck})$, the cryptographic pseudorandom number generator $G_1$ fools
$\mathcal{S}_2$. Let $h$ be the one-way permutation used by the
pseudorandom number generator $G_1$, and assume $h$ us computable in time $n^c$ for some
$c > 0$. Given $h$, $G_1$ is defined as follows: let $x,y \in
\{0,1\}^{\frac{n^{\delta}}{2}}$, and let $(x,y)$ be $G_1$'s
seed (notice that $|(x,y)| = n^\delta$), then
	\[
	G_1(x, y) = \left( f^n(x) \odot y, f^{n-1}(x) \odot y, \dots, f(x) \odot y \right),
	\]
where $x \odot y = \sum_i x_i y_i \mod 2$. There are $O(n^{2})$
applications of $h$, so $G_1$ runs in time $O(n^{2 +
c\delta})$. So, for $k \ge 2 + c\delta$, $G_1$ fools
$\mathcal{S}_2$.

At this point we're almost done. As in Lemma
\ref{one-way-to-nash} assume, by contradiction, that the
assumptions in the theorem hold but $(G_1, G_2)$ is not a
$\varepsilon$-Nash equilibrium for some $\varepsilon > 0$. This
implies that at least one of the two players can improve his
expected payoff by more than $\varepsilon$ by switching to some
other strategy. First consider P2, and assume there is a
strategy $S_2 \in \mathcal{S}_2$ such that
	\[
	\e{U(G_1, S_2)} > \e{U(G_1, G_2)} + \varepsilon.
	\]
As in Lemma \ref{one-way-to-nash} this implies that $S_2$ would
be a test for the cryptographic pseudorandom number generator $G_1$, contradicting the
fact that $G_1$ fools $\mathcal{S}_2$. Similarly, assume P1 has
a strategy $S_1 \in \mathcal{S}_1$ such that
	\[
	\e{U(S_1, G_2)} > \e{U(G_1, G_2)} + \varepsilon.
	\]
Again, $S_1$ can easily be made into a test for $G_2$,
contradicting the fact that $G_2$ fools $\mathcal{S}_1$.
\end{proof}

\section{Infinite Play} \label{sec:infinite}
We now consider an infinitely repeated game of Matching
Pennies, and show that, if utilities are time discounted, we
can always achieve a $\varepsilon$-Nash equilibria using a
large enough (but finite) amount of random coins. First we
determine the least amount of randomness required to achieve a
$\varepsilon$-Nash equilibrium in the general, i.e.
computationally unbounded, case.

\begin{lemma} \label{lemma:infinite1}
For all discount factors $\delta \in (0,1)$ and all
$\varepsilon > 0$, there is an $\varepsilon$-Nash equilibrium
in which the players use only $n$ random bits, for $n >
\frac{\log \varepsilon (1- \delta)}{\log \delta}$.
\end{lemma}

\begin{proof}
Given $\delta, \varepsilon > 0$ consider the following
strategies: both players play the Nash equilibrium strategy for
the first $n$ rounds. After this P1 will always play $H$, while
P2 will play $H$ and $T$ alternatively. The overall expected
payoff is 0. However, after round $n$, both players could
switch to a strategy that always wins, achieving a total
expected payoff of $0 + \sum_{t = n}^{\infty} \delta^t =
\frac{\delta^n}{1 - \delta}$. To ensure that our strategies are
indeed a $\varepsilon$-Nash equilibrium we just need to make
sure that
	\[
	0 > \frac{\delta^n}{1 - \delta} -\varepsilon.
	\]
Rearranging and taking logarithms we get $n > \frac{\log
\varepsilon (1- \delta)}{\log \delta}$.
\end{proof}

Now we consider players' whose strategies are families of
polynomial size Boolean circuits (as in Section \ref{sec:eff}),
and assume one-way functions exist. We first give a version of
Lemma \ref{lemma} for time discounted utilities on a finite
number of rounds.

\begin{lemma} \label{lemma:discounted}
Assume one-way functions exist, and let $G = (g^1, \dots, g^n)$
be the strategy corresponding to the output of a cryptographic
pseudorandom number generator. Let $S = (s^1, \dots, s^n)$ be any strategy. For all
$\delta \in (0,1)$, $k > 0$ and for sufficiently large $n$
	\[
	\Big| \e{\sum_{t = 1}^n \delta^t h(g^t, s^t)} \Big| \le n^{-k} \text{ and } \Big| \e{\sum_{t = 1}^n \delta^t h(s^t, g^t)} \Big| \le n^{-k}.
	\]
\end{lemma}
\begin{proof}
Again we give the proof only for the first inequality. Assume,
by contradiction, that $| \e{\sum_{t = 1}^n \delta^t h(g^t,
s^t)} | > n^{-k}$ for some $\delta$, $k$ and infinitely many
$n$'s. Consider the random variables $A_1(c_1^1, c_2^1), \dots,
A_n(c_1^n, c_2^n)$, defined as $A_t(c_1^t, c_2^t) = 1$ if P1
wins round $t$ when playing according to $(C_1, C_2)$ and 0
otherwise. Let $A(C_1, C_2) = \sum_t \delta^t A_t(c_1^t,
c_2^t)$, so that $\e{A(G, S)} \ge |\e{\sum_{t = 1}^n \delta^t
h(g^t, s^t)}| > n^{-k}$. This implies that there is a $t$ such
that $\delta^t \e{A_t(g^t, s^t)} > n^{-k-1}$, which implies
$\e{A_t(g^t, s^t)} > n^{-k-1}$. Fix that $t$. As in Lemma
\ref{lemma}, consider the test $T$ that, given a sequence of
plays $x$, generates a play $s$ from $S$ and outputs 1 if P1
wins round $t$ and outputs 0 otherwise.

When $x$ is drawn uniformly at random, $\pr{T(U_n) = 1} = 1/2$.
On the other hand, when $x$ is $G$'s output,
	\[
	\pr{T(G(U_{l(n)})) = 1} = \e{A_t} > n^{-k-1}.
	\]
Now
	\[
	\Big| \pr{T(G(U_{l(n)})) = 1} - \pr{T(U_n) = 1} \Big| > \frac{1}{2} - \frac{1}{n^{k+1}} = \frac{1}{n^c},
	\]
for $c > - \frac{\log(1/2 - n^{-1-k})}{\log n} \ge 0$,
contradicting the assumption that $G$ is a pseudorandom number generator.
\end{proof}

Using the above Lemma we can show that, for all discount
factors, we can greatly reduce the amount of random coins
needed to get an $\varepsilon$-Nash equilibrium.

\begin{lemma} \label{lemma:infinite2}
For all discount factors $\delta \in (0,1)$, all $\varepsilon >
0$ and all $\xi > 0$, there is a $n^{-k}$-Nash equilibrium in
which players use only $n^\xi$ random coins, for sufficiently
large $n$'s.
\end{lemma}

\begin{proof}
As in the proof of Lemma \ref{lemma:infinite1} we consider the
following strategy for both players: for the first $n$ rounds
play the output of a cryptographic pseudorandom number generator $G$, with seed length
$n^\xi$. Thereafter P1 will always play $H$, while P2 will
alternate between $H$ and $T$. Pick any $k > 0$, we now show
that this is a $n^{-k}$-Nash equilibrium. By Lemma
\ref{lemma:discounted} we can pick $k' = k/2$ such that the
expected utility in the first $n$ rounds lies in the interval
$[-n^{-k'}, n^{-k'}]$. To ensure that this is a $n^{-k}$-Nash
equilibrium we just need to show that
	\[
	-n^{-k'} > n^{-k'} + \frac{\delta^n}{(1 - \delta)} - n^{-k},
	\]
or
	\[
	-2n^{-k'} + n^{-2k'} > \frac{\delta^n}{(1 - \delta)}.
	\]
Now, for any $c > 0$, if we set $k' = \frac{\log\left(
(\sqrt{c+1} - 1)/c \right)}{\log n}$, then the left hand side
of the above inequality is $c$, so that it always holds for
sufficiently large $n$'s.
\end{proof}

Thus, given any $n$ that satisfies the conditions in Lemma
\ref{lemma:infinite1}, there can be a $n^{-k}$-Nash equilibrium
using $n^{\delta}$ coins, for any $\delta > 0$. To see this,
consider an $m$ sufficiently large so that Lemma
\ref{lemma:infinite2} holds and pick $\xi =  \delta \frac{\log
n}{\log m}$.

\section{Conclusions}
We have shown how, in a simple setting, reducing the amount of
randomness available to players affects Nash equilibria. In
particular, if we make no computational assumptions on the
players, there is a direct tradeoff between the amount of
randomness and the approximation to a Nash equilibrium we can
achieve. If, instead, players are bound to run in polynomial
time, we can get very close to a Nash equilibrium with only
$n^{\delta}$ random coins, for any $\delta > 0$.

Some directions for future research include:
	\begin{itemize}
	\item Is it possible to extend Lemma
\ref{lemma:unbounded1} to $m$ player games, for $m > 2$?
Notice that the strategy used in that proof does not
generalize to this setting.
	\item Under what circumstances is it possible to
further reduce the amount of randomness available (say to
$O(\log n)$ for both players)?
	\item Is it possible to extend these results to general zero-sum games or even non zero-sum games?
	\end{itemize}

\paragraph{Acknowledgments} We wish to thank Tai-Wei Hu, Peter Bro Miltersen, Rahul Santhanam and Rakesh Vohra for fruitful discussions.

\end{document}